\documentclass{llncs}
\usepackage{mathtools}
\usepackage{enumerate}
\usepackage{hyperref}

\usepackage[utf8]{inputenc}
\usepackage[english]{babel}
\usepackage{enumerate}
\usepackage{url}
 
\newcommand{\bfe}[1]{\begin{bfseries}\emph{#1}\end{bfseries}\index{#1}}

\newcommand{\myra}{\mbox{$\:\rightarrow\:$}}

\newcommand{\sse}{\mbox{$\:\subseteq\:$}}

\newcommand{\LL}{\mbox{$\ldots$}}

\newcommand{\NI}{\noindent}
\newcommand{\HB}{\hfill{$\Box$}}

\newcommand{\III}{\vspace{3 mm}}
\newcommand{\II}{\vspace{2 mm}}

\newcommand{\szkew}[1]{\relax \setbox0=\hbox{\kern -24pt $\displaystyle#1$\kern 0pt }%
\box0}
{\catcode`\@=11 \global\let\ifjusthvtest@=\iffalse}

\newcounter{oldmycaption}

\usepackage{url}
\usepackage{hyperref}
\usepackage{graphicx}
\usepackage{latexsym}
\usepackage{amssymb}
\usepackage{amsmath}

\begin{document}

\title{Self-Stabilization Through the Lens of Game Theory}
\author{Krzysztof R. Apt\inst{1,2} 
\and Ehsan Shoja\inst{3}}
\institute{%
CWI, Amsterdam, The Netherlands
\and
MIMUW, University of Warsaw, Warsaw, Poland
\and
Sharif University of Technology, Tehran, Iran
}

\maketitle

\begin{abstract}
  In 1974 E.W.~Dijkstra introduced the seminal concept of
  self-stabilization that turned out to be one of the main approaches
  to fault-tolerant computing. We show here how his three solutions
  can be formalized and reasoned about using the concepts of game
  theory.  We also determine the precise number of steps needed to
  reach self-stabilization in his first solution.
\end{abstract}

\section{Introduction}

In 1974 Edsger W.~Dijkstra introduced in a two-page article
\cite{Dij74} the notion of self-stabilization.  The paper was
completely ignored until 1983, when Leslie Lamport stressed its
importance in his invited talk at the ACM Symposium on Principles of
Distributed Computing (PODC), published a year later as \cite{Lam84}.
Things have changed since then.  According to Google Scholar
Dijkstra's paper has been by now cited more than 2300 times.  It
became one of the main approaches to fault tolerant computing.  An
early survey was published in 1993 as \cite{Schneider:1993}, while the
research on the subject until 2000 was summarized in the book
\cite{Dol00}.  In 2002 Dijkstra's paper won the PODC influential paper
award (renamed in 2003 to Dijkstra Prize). The literature on the
subject initiated by it continues to grow.  There are annual
Self-Stabilizing Systems Workshops, the 18th edition of which took
part in 2016.

The idea proposed by Dijkstra is very simple. Consider a distributed
system viewed as a network of machines. Each machine has a local state
and can change it autonomously by inspecting its local state and the
local states of its neighbours. Some global states are identified as
\emph{legitimate}.  A distributed system is called self-stabilizing if it
satisfies the following three properties (the terminology is from \cite{AG93}):

\begin{description}
\item[closure:] starting from an arbitrary global state, the system is guaranteed to
reach a legitimate state,

\item[stability:] once a legitimate state is reached, the system remains in it forever,

\item[fairness:] in every infinite sequence of moves every machine is selected infinitely often.
\end{description}

Dijkstra proposed in \cite{Dij74} three solutions to
self-stabilization in which, respectively, $n$, four and three state
machines were used, where $n$ is the number of machines.  The proofs
were provided respectively in \cite{EWD:EWD391} (republished as
\cite{Dijkstra1982}), \cite{EWD:EWD392} and \cite{EWD:EWD396}
(republished with small modifications as \cite{Dij86}).  In his
solutions a legitimate state is identified with the one in which
exactly one machine can change its state.

In this paper we show how Dijkstra's solutions to self-stabilization
can be naturally formulated using the standard concepts of strategic
games, notably the concept of an improvement path.  Also we show how
one can reason about them using game-theoretic terms.  We focus on
Dijkstra's first solution but the same approach can be adopted to
other solutions.

The connections between self-stabilization and game theory were
noticed before. We discuss the relevant references in the final section.
The analysis of the original Dijkstra's solutions using game theory is
to our knowledge new.

This paper connects two unrelated areas, each of which has developed
its own well-established notation and terminology. To avoid possible
confusion, let us clarify that in what follows $S_i$ denotes a set of
strategies of a player in a strategic game, while the letter $S$
denotes a variable in a solution to the self-stabilization problem.
Further, the notion of a state in the self-stabilization refers to the
range of a variable and not to an assignment of values to all
variables, as is customary in the area of program semantics.

\section{Preliminaries}
\label{sec:prem}

A \bfe{strategic game} $\mathcal{G}=(S_1, \ldots, S_n,$
$p_1, \ldots, p_n)$ for $n > 1$ players consists of a non-empty set
$S_i$ of \bfe{strategies} and a \bfe{payoff function}
$p_i : S_1 \times \cdots \times S_n \myra \mathbb{R}$, for each player
$i$.  We denote $S_1 \times \cdots \times S_n$ by $S$, call each
element $s \in S$ a \bfe{joint strategy} and abbreviate the sequence
$(s_{j})_{j \neq i}$ to $s_{-i}$. Occasionally we write
$(s_i, s_{-i})$ instead of $s$.  We call a strategy $s_i$ of player
$i$ a \bfe{best response} to a joint strategy $s_{-i}$ of his
opponents if for all $ s'_i \in S_i$,
$p_i(s_i, s_{-i}) \geq p_i(s'_i, s_{-i})$.  A joint strategy $s$ is
called a \bfe{Nash equilibrium} if each $s_i$ is a best response to
$s_{-i}$. (In the literature these equilibria are often called
\emph{pure Nash equilibria} to distinguish them from Nash equilibria
in mixed strategies. The latter ones have no use in this paper.)

Further, we call a strategy $s'_i$ of player $i$ a \bfe{better
  response} given a joint strategy $s$ if
$p_i(s'_i, s_{-i}) > p_i(s_i, s_{-i})$.  We call $s \to s'$ an
\bfe{improvement step} (abbreviated to a \bfe{step}) if
$s' = (s'_i, s_{-i})$ for some better response $s'_i$ of player $i$
given $s$.  So $p_i(s') > p_i(s)$.

An \bfe{improvement path} is a maximal sequence
\[
s^1 \to s^2 \to \LL \to s^k \to \LL
\]
such that each $s^i \to s^{i+1}$ is an improvement step.

In the next section we consider specific strategic games on directed
graphs.  Fix a finite directed graph $G$.  We say that a node $j$ is a
\bfe{neighbour} of the node $i$ in $G$ if there is an edge $j \to i$
in $G$.  Let $N_i$ denote the set of all neighbours of node $i$ in the
graph $G$.  We now consider a strategic game in which each player is a
node in $G$.  Fix a non-empty set of strategies $C$ that we call
\bfe{colours}.

We divide the players in two categories: those who play a coordination game
and those who play an anti-coordination game.
More specifically,

\begin{itemize}
\item the players are the nodes of $G$,
  
\item the set of strategies of player (node) $i$ is a set of colours $A(i)$ such that   $A(i) \sse C$,

\item if the player plays the coordination game, then his
 payoff function is defined by
\[
p_i(s) = |\{j \in N_i \mid s_i = s_j\}|,
\]

\item if the player plays the anti-coordination game, then his payoff
  function is defined by
\[
p_i(s) = |\{j \in N_i \mid s_i \neq s_j\}|.
\]
\end{itemize}

So each node simultaneously chooses a colour and the payoff to the
player who plays the \emph{coordination game} is the number of its neighbours
that chose its colour, while the payoff to the player who plays the
\emph{anti-coordination game} is the number of its neighbours that chose a
different colour.

The games on directed graphs in which all players were playing the
coordination game were studied in \cite{ASW16}. Corresponding games on
undirected graphs were considered in \cite{AKRSS2017} and on weighted
undirected graphs in \cite{rahn:schaefer:2015}. In turn, the games in
which some players played the coordination game while other players
played the anti-coordination game were studied (in a more general
context of weighted hypergraphs) in \cite{SW17}.  If the underlying
(weighted) graph is undirected the game always has a Nash equilibrium,
which is not the case if the graph is directed. The absence of Nash
equilibria is crucial in the context of this paper.

We now move on to the subject of this paper and introduce the following concepts
concerning improvement paths.

\begin{definition}
Fix a strategic game.

\begin{itemize}
\item A joint strategy is \textbf{legitimate} if exactly one player does not play a best response in it.

\item An improvement path ensures

  \begin{itemize}
  \item \textbf{closure} if some joint strategy in it is legitimate, 

  \item \textbf{stability} if the successors of the legitimate joint strategies in it are legitimate, 

  \item \textbf{fairness} if every player is selected in it infinitely often, 

  \item \textbf{self-stabilization} (\textbf{in $k$
    steps}) if every player is selected in it infinitely often and from
  a certain point (\textbf{after $k$ steps}) each joint strategy in it is
  legitimate.

  \end{itemize}

\item A game \textbf{admits closure/stability/fairness}
if it is ensured by every improvement path in it.

\item A game \textbf{admits self-stabilization} (\textbf{in $k$ steps}) 
if it is ensured by every improvement path in it (\textbf{in $k$ steps}).

\end{itemize}

\end{definition}

For a more refined analysis we shall need the concept of a scheduler.

\begin{definition}
\mbox{} 

\begin{itemize}

\item A \textbf{scheduler} is a function $f$ that given a joint
  strategy $s$ that is not a Nash equilibrium and a player $i$ who
  does not hold in $s$ a best response selects a strategy $f(s,i)$ for
  $i$ that is a better response given $s$.

\item Consider a scheduler $f$. 
An improvement path 
\[
s^1 \to s^2 \to \LL \to s^k \to \LL,
\] 
is \textbf{generated by $f$} if for each $k \geq 1$, if $s^k$ is not a Nash equilibrium, then for some $i \in \{1, \LL, n\}$, 
$s^{k+1} = (f(s^k,i), s^k_{-i})$.

\item A scheduler $f$ \textbf{ensures self-stabilization} (\textbf{in $k$ steps}) 
if every improvement path generated by it ensures self-stabilization (\textbf{in $k$ steps}).
\end{itemize}

\end{definition}

So a game admits self-stabilization (in $k$ steps) if every scheduler
ensures self-stabilization (in $k$ steps).  Schedulers in the context
of strategic games were extensively considered in \cite{AS15}, though
they selected a player and not his strategy. The ones used here
correspond in the terminology of \cite{AS15} to the state-based
schedulers.

\section{Dijkstra's first solution}
\label{sec:1}

We start by recalling the first solution to the self-stabilization
problem given in \cite{Dij74}.  We assume a directed ring
of $n$ machines, each having a local variable and a program.  The
variables assume the values from the set $\{0, \LL, k-1\}$, where
$k \geq n$ and $\oplus$ stands for addition modulo $k$.
Each program consists of a single rule of the form
\[
P \to A
\]
where $P$ is a condition, called a \emph{priviledge}, on the local
variables of the machine and its predecessor in the ring, and $A$ is
an assignment to the local variable.  The variable of a considered
machine is denoted by $S$ and the variable of its predecessor by $L$.

The program for machine 1 is given by the rule
\[
L = S \: \to \: S := S \oplus 1
\]
and for the other machines by the rule
\[
L \neq S \: \to \: S := L.
\]

One assumes that each time a machine is selected, its priviledge is
true.  Dijkstra proved in \cite{Dijkstra1982} (that originally
appeared as \cite{EWD:EWD391}) that starting from an arbitrary initial
situation any sequence of machine selections leads to a situation in
which
\begin{itemize}
\item exactly one priviledge is true,

\item this property remains true forever.
\end{itemize}
Moreover, every machine is selected in this sequence infinitely often.

In the terminology introduced in the Introduction
the above system of machines is self-stabilizing.

We can model the above solution by means of 
the following strategic game $\cal{G}$ on a directed ring involving $n$ players:

\begin{itemize}
\item each player has the same set $C$ of strategies (called colours), where $|C| \geq 2$,

\item exactly one player plays the anti-coordination game on the ring,

\item all other players play the coordination game  on the ring.
\end{itemize}

To fix notation we assume that it is player 1 who plays the anti-coordination game.
So the payoff functions are simply:
\[
p_1(s) := \begin{cases}
    0 & \mathrm{if}\ s_1 = s_n \\
    1 & \mathrm{otherwise}
\end{cases}
\]
and for $i \neq 1$
\[
p_i(s) := \begin{cases}
    0 & \mathrm{if}\ s_i \neq s_{i-1} \\
    1 & \mathrm{otherwise}
\end{cases}
\]

We arrange the colours in $C$ in a cyclic order and given a colour $c$
we denote its successor in this order by $c'$.  
The following result provides a game-theoretic account of the above
solution to the self-stabilization problem.

\begin{theorem} \label{thm:0} Consider the game $\cal{G}$. Suppose
  that $n \geq 3$ and $|C| \geq n$.  Let $f$ be a scheduler such that
\[
  \begin{array}{l}
\mbox{$f(s,1) = s'_1$.}
  \end{array}
\]
Then $f$ ensures self-stabilization in $\cal{G}$. 
\end{theorem}

Thus the only restriction on the scheduler $f$ is that for player 1 it
selects the next colour in the cyclic order on $C$ (as $s'_1$ denotes
the successor of $s_1$).  

\begin{proof}
  There is a 1-1 correspondence between the maximal sequences of moves of the machines in Dijkstra's solution
and the improvement paths generated by the schedulers satisfying the stated condition.
\HB
\end{proof}

We shall return to the above result in Section \ref{sec:analysis}.  It
is useful to point out why we did not incorporate the specific choice
of the strategies into the payoff functions and used a scheduler
instead.  This alternative would call for selecting $\{0, \LL, k-1\}$
as the set of strategies for each player and using the following
payoff function for player 1, where $\oplus$ stands for addition
modulo $k$:
\[
p_1(s) := \begin{cases}
    0 & \mathrm{if}\ s_1 \neq s_n \oplus 1\\
    1 & \mathrm{otherwise}
\end{cases}
\]
However, the resulting game would then correspond to a setup
in which the program for machine 1 is
\[
S \neq L \oplus 1\: \to \: S := L \oplus 1.
\]

Moreover, the resulting game does not admit self-stabilization (and a
fortiori the resulting programs for the machines do not form a
solution for self-stabilization). Indeed, assume three players and
$k = 3$, so that the strategies of the players are 0, 1, 2.  Then the
following infinite improvement path does not ensure closure:
\[
  \begin{array}{l}
(200 \to 220 \to 120 \to 122 \to 112 \to 012 \to 011 \to 001 \to 201 \to)^{*},
  \end{array}
\]
where each joint strategy is displayed as a string of three numbers
from $\{0,1,2\}$ and $^*$ stands for the infinite repetition of the
exhibited prefix of an improvement path.

\section{Dijkstra's three-state solution}

Next we discuss Dijkstra's three-state solution to the
self-stabilization problem.  We follow here the presentation he
gave in \cite{Dij86}, where he provided a particularly elegant
correctness proof.

There are $n$ machines arranged in an undirected ring, the first one
called the \emph{bottom} machine, the last one called the \emph{top}
machine, and the other machines called \emph{normal}.

The condition of each rule is now on the local variables of the machine
and its two neighbours.  The variable of a considered machine is
denoted by $S$, of its left neighbour by $L$ and of its right
neighbour by $R$.  All variables range over the set $\{0,1,2\}$ and
$\oplus$ stands for addition modulo 3.

The program for the bottom machine is given by the rule
\[
S \oplus 1 = R \: \to \: S := S \oplus 2,
\]
for each normal machine by the rule
\[
L = S \oplus 1 \lor S \oplus 1 = R \: \to \: S := S \oplus 1,
\]
and for the top machine by the rule
\[
L = R \land S \neq R \oplus  1 \: \to \: S := R \oplus 1.
\]
Dijkstra proved that the above system of machines is self-stabilizing.

This solution can be represented and reasoned about using strategic
games, though these games are not anymore coordination or
anti-coordination games.  First note that, in contrast to the case of
Dijkstra's first solution, this solution cannot be modeled using
strategic games with 0/1 payoffs.  To see it assume $n = 3$ and
consider the global state of the system described by $(2,1,0)$. Then
the priviledge of machine 2 is true, since $L = S \oplus 1$, as
$2 = 1 \oplus 1$.  After machine 2 is selected the global state
changes to $(2,2,0)$. In this state the priviledge of machine 2 is
again true, since $S \oplus 1 = R$, as $2 \oplus 1 = 0$. So in the
improvement path of the corresponding strategic game player 2 can be
selected twice in succession. This can be modelled only using at least
three payoff values.

To capture such a possibility we need to analyze when a machine can be
selected twice in succession. This can happen when successively
$L = S \oplus 1$ and $S \oplus 1 = R$ are true or successively
$S \oplus 1 = R$ and $L = S \oplus 1$ are true. Taking into account
the action of the assignment $S := S \oplus 1$ the first possibility
means that initially $L = S \oplus 1 \land S \oplus 2 = R$ is true and
the second possibility that initially
$S \oplus 1 = R \land L = S \oplus 2$ is true. These two options can
be rewritten as
$S \oplus 1 \in \{L, R\} \land S \oplus 2 \in \{L, R\}$.

To complete this analysis note that a machine can be selected only
once in succession, when initially
$L = S \oplus 1 \land S \oplus 2 \neq R$ is true or
$S \oplus 1 = R \land L \neq S \oplus 2$ is true, which can be
rewritten as
$S \oplus 1 \in \{L, R\} \land S \oplus 2 \not\in \{L, R\}$.

Translating it into a game-theoretic notation that uses indices we are
brought into the following strategic game $\cal{G}$ for $n$ players.
Each player has $\{0,1,2\}$ as the set of strategies. The payoff
functions are defined as follows, where we assume that player 1
corresponds to the bottom machine and player $n$ to the top machine:
{
\setlength{\abovedisplayshortskip}{-10pt}
\begin{align*}
p_1(s) &:=
\begin{cases} 
0 &  \text{if } s_1 \oplus 1= s_2 \\
1 &  \text{otherwise}\\
\end{cases}\\[10pt]
\shortintertext{for $1 < i < n$} &\nonumber \\
p_i(s) &:=
\begin{cases} 
0 &  \text{if } s_i \oplus 1 \in\{s_{i-1}, s_{i+1}\} \land s_i\oplus 2 \in \{s_{i-1}, s_{i+1}\} \\
1 &  \text{if } s_i \oplus 1 \in\{s_{i-1}, s_{i+1}\} \land s_i\oplus 2 \not\in \{s_{i-1}, s_{i+1}\} \\
2 & \text{otherwise}\\
\end{cases}\\[5pt]
p_n(s) &:=
\begin{cases} 
0 &  \text{if } s_1 = s_{n-1} \land s_n \neq s_1 \oplus 1  \\
1 &  \text{otherwise} 
\end{cases}
\end{align*}
}

Dijkstra's result concerning the above system of three-state machines
is captured by the following theorem.

\begin{theorem} \label{thm:3} Consider the above game
  $\cal{G}$. Suppose that $n \geq 3$.  Let $f$ be a scheduler such
  that
\[
  \begin{array}{l}
\mbox{$f(s, 1) = s_1 \oplus 2$}, \\
\mbox{$f(s, i) = s_i \oplus 1$, where $1 < i < n$}, \\
\mbox{$f(s, n) = s_1 \oplus 1$}.
  \end{array}
\]
Then $f$ ensures self-stabilization in $\cal{G}$.
\end{theorem}

\begin{proof}
Every maximal sequence of moves of the machines in Dijkstra's three-state solution corresponds to
an improvement path generated by a scheduler satisfying the stated conditions.
Conversely, every improvement path generated by a scheduler satisfying the stated conditions corresponds
to a maximal sequence of moves of the machines in Dijkstra's three-state solution with each improvement step
that results for a player $i$ in the payoff increase by 2 mapped to two consecutive moves of machine $i$.
\HB
\end{proof}

\section{A four-state solution}

Finally, we consider a four-state solution. Instead of Dijkstra's
solution that uses two Boolean variables per machine we consider a
modified solution due to \cite{Ghosh1993} that uses per machine a
single variable that can take four values. We assume the set up and
terminology of the previous section, with the following differences.

The variable of machine 1 now ranges over $\{1,3\}$, of machine $n$
over $\{0,2\}$.  and all other variables range over $\{0,1,2,3\}$.
Further, $\oplus$ stands now for addition modulo 4.

The program for the bottom machine is given by the rule
\[
S \oplus 1 = R \: \to \: S := S \oplus 2,
\]
for each normal machine by the rule
\[
L = S \oplus 1 \lor S \oplus 1 = R \: \to \: S := S \oplus 1,
\]
and for the top machine by the rule
\[
L = S \oplus 1 \: \to \: S := S \oplus 2.
\]

Following the considerations of the previous section this solution can
be modeled by the following strategic game $\cal{G}$ for $n$ players.
The sets of strategies are as follows: for player 1: $\{1,3\}$, for
player $n$: $\{0,2\}$, and for all other players: $\{0,1,2,3\}$.

The payoff functions are defined as follows, where we assume that
player 1 corresponds to the bottom machine and player $n$ to the top
machine: 
{ 
\setlength{\abovedisplayshortskip}{-10pt}
\begin{align*}
p_1(s) &:=
\begin{cases} 
0 &  \text{if } s_1 \oplus 1 = s_2  \\
1 &  \text{otherwise}\\
\end{cases}\\[10pt]
\shortintertext{for $1 < i < n$} &\nonumber \\
p_i(s) &:=
\begin{cases} 
0 &  \text{if } s_i \oplus 1 \in\{s_{i-1}, s_{i+1}\} \land s_i\oplus 2 \in \{s_{i-1}, s_{i+1}\} \\
1 &  \text{if } s_i \oplus 1 \in\{s_{i-1}, s_{i+1}\}) \land s_i\oplus 2 \not\in \{s_{i-1}, s_{i+1}\} \\
2 & \text{otherwise}\\
\end{cases}\\[5pt]
p_n(s) &:=
\begin{cases} 
0 &  \text{if } s_n \oplus 1 = s_{n-1} \\
1 &  \text{otherwise}
\end{cases}
\end{align*}
}

The reason for using three values in the payoff functions $p_i$, where
$1 < i < n$, is as in the previous section.  The corresponding result
concerning self-stabilization of the above system of four-state machines
is now captured by the following game-theoretic theorem.

\begin{theorem} \label{thm:4} 
Consider the above game $\cal{G}$. Suppose
  that $n \geq 3$.
Let $f$ be a scheduler such that
\[
  \begin{array}{l}
\mbox{$f(s, 1) = s_1 \oplus 2$}, \\
\mbox{$f(s, i) = s_i \oplus 1$, where $1 < i < n$}, \\
\mbox{$f(s, n) = s_n \oplus 2$}.
  \end{array}
\]
Then $f$ ensures self-stabilization in $\cal{G}$.
\end{theorem}

\begin{proof}
  The same as the proof of Theorem \ref{thm:3}.
\HB
\end{proof}

\section{A game-theoretic analysis of the first solution}

\label{sec:analysis}

We now analyze in detail the strategic game $\cal{G}$ introduced in Section
\ref{sec:1} with the aim of proving a stronger result about the first
solution to self-stabilization.  We begin with the following
observation.

\begin{note} \label{not:1}
The game $\cal{G}$ admits no Nash equilibria.
\end{note}
\begin{proof}
  Suppose otherwise. Let $s$ be a Nash equilibrium of $\cal{G}$.  Then
  every player $i \neq 1$ holds in $s$ the colour of its
  predecessor. Hence all players hold in $s$ the same colour, in
  particular players $1$ and $n$. But then player $1$ does not hold in
  $s$ a best response, which yields a contradiction.
\HB
\end{proof}

\begin{corollary} \label{cor:2}
The game $\cal{G}$ admits stability.
\end{corollary}

\begin{proof}
  Suppose $s \to s'$ is an improvement step in the game $\cal{G}$ and
  that $s$ is legitimate. Then by the definition of the game either
  $s'$ is legitimate or is a Nash equilibrium. So the claim follows by
  Note \ref{not:1}.  
\HB
\end{proof}

We shall use below the following observation.

\begin{note} \label{not:3}
  Consider a coordination game on a chain of $n$ players in which each player
has the same set of strategies. Then all improvement paths in this game
are of length $\leq \frac{n (n-1)}{2}$.
Further, improvement paths of length $\frac{n (n-1)}{2}$ exist.
\end{note}
\begin{proof}
  Suppose the chain is $1 \to 2 \to \LL \to n$. Consider an improvement
  path $\xi$.  Each player $i$ can adopt in $\xi$ at most $i-1$
  colours, namely the strategies held by his predecessors in the
  chain. So each player $i$ can be involved in at most $i-1$
  improvement steps. Consequently the length of $\xi$ is bound by
  $\sum_{i = 1}^n (i - 1) = \frac{n (n-1)}{2}$.

  To establish the second claim take an initial joint strategy $s$ in
  which all colours differ. Then the required number of steps is
  achieved by scheduling the players in the `rightmost first' order,
  so
\[
(n), (n-1, n), (n-2, n-1, n), \LL, (2, 3, \LL, n),
\]
where to increase readability we separated the consecutive phases
using brackets.

\HB
\end{proof}

\begin{theorem} \label{thm:fairness}
The game $\cal{G}$ admits fairness.
\end{theorem}

\begin{proof}
  Consider an improvement path $\xi$.  We first prove that player $1$
  is infinitely often selected in $\xi$.  Suppose otherwise.  By Note
  \ref{not:1} $\xi$ is infinite, so from some moment on player $1$ is
  never selected in the infinite suffix $\phi$ of $\xi$. Break the
  ring by removing the link between players $n$ and $1$ and consider
  the resulting coordination game on the chain
  $1 \to 2 \to \LL \to n$.  Then $\phi$ is an infinite improvement
  path in this game, which contradicts Note \ref{not:3}.

  Note now that if some player $i$ is finitely often selected in
  $\xi$, then so is its successor.  Together with the above
  conclusion this implies successively that players $n, n-1, \LL, 2$
  are infinitely often selected in $\xi$.
\HB
\end{proof}

So to prove that  $\cal{G}$ admits self-stabilization we only need to check
that it admits closure. However, this holds only for games with two or three players.
In fact, we have the following result.

\begin{theorem} \label{thm:1}
Consider the game $\cal{G}$.
\begin{enumerate}[(i)]
\item If $n = 2$ then $\cal{G}$ admits self-stabilization in $0$ steps.

\item If $n = 3$ then $\cal{G}$ admits self-stabilization in $2$ steps.

\item If $n > 3$ then $\cal{G}$ does not admit self-stabilization.

\end{enumerate}
\end{theorem}

\begin{proof}
For simplicity we view each joint strategy as a string over the set of colours
that we denote by the initial letters of the alphabet. Different letters
stand for different colours.
  
\NI
$(i)$ In this case every joint strategy is legitimate.
\II

\NI
$(ii)$ For brevity we say that a joint strategy $s$ is an
\emph{$i$-strategy}, where $0 \leq i \leq 2$, if exactly $i$ players
hold in $s$ a best response.  The only $0$-strategy is of the form
$aba$. We reach from it in one step a $1$-strategy $cba$ (assuming $|C| > 2$)
or a $2$-strategy $bba$, $aaa$ or $abb$.

So consider now an arbitrary $1$-strategy.  If it is player 1 who
plays the best response, then $s$ is of the form $acb$ (so in this
case $|C| > 2$). Then the only possible improvement steps are
$acb \to aab$ or $acb \to acc$. In both cases we reach a $2$-strategy
in one step.

If it is player 2 who plays the best response, then $s$ is of the form
$aaa$ or $aab$, which contradicts the fact that $s$ is a $1$-strategy.
Finally, if it is player 3 who plays the best response, then $s$ is of
the form $baa$ or $aaa$, which also contradicts the fact that $s$ is a
$1$-strategy.

We conclude that a legitimate joint strategy is always reached in at most 2 steps.
\II

\NI
$(iii)$
Assume that $n > 3$.
Then the following infinite improvement path does not ensure closure:
\[
  \begin{array}{l}
(bba^{n-4}ab \to aba^{n-4}ab \to aba^{n-4}aa \to^{*} abb^{n-4}ba \to \\
\phantom{(}aab^{n-4}ba \to bab^{n-4}ba \to bab^{n-4}bb \to^{*} baa^{n-4}ab \to)^{*},
  \end{array}
\]
where each inner $^*$ stands for an appropriate sequence of $n-4$
improvement steps, while the outer $^*$ stands for the infinite
repetition of the exhibited prefix of an improvement path.  
\HB
\end{proof}

The above result explains the need for a scheduler.
As before we assume a cyclic order on the set of colours and denote
the successor of colour $c$ by $c'$.  The following result improves
upon Theorem \ref{thm:0}.  The differences are discussed after the
proof.

\begin{theorem} \label{thm:2} 
Consider the game $\cal{G}$. Suppose
  that $n \geq 3$ and $|C| \geq n-1$.  Let $f$ be a scheduler such
  that 
\[
  \begin{array}{l}
\mbox{$f(s,1) = s'_1$.}
  \end{array}
\]
Then $f$ ensures self-stabilization in $\cal{G}$ in 
$\frac12 (3 n + 1)(n - 2)$ steps.
\end{theorem}

\begin{proof}
We split the proof in two parts. The slightly unusual naming
of joint strategies in Part 1 will become clear in Part 2.
\II

\NI
\emph{Part 1: self-stabilization.}

Consider an improvement path $\xi$ generated by the
scheduler $f$ that starts in a joint strategy $s$.  Call a joint
strategy \emph{lean} if the players $2, \LL, n$ hold in it at most
$n-2$ different colours.  
We now establish a number of claims about $\xi$.  
\III

\NI
\emph{Claim 1.} A lean joint strategy appears in $\xi$.
\II

\NI
\emph{Proof}. 
By Theorem \ref{thm:fairness}
eventually some player $i \in \{3, \LL, n\}$ is selected in $\xi$. 
The resulting joint strategy becomes then lean.  
\HB
\II

Let $s''$ be the first lean joint strategy in $\xi$.  Call a colour
\emph{fresh} in $\xi$ if it is not held in $s''$ by any player
$i \neq 1$.  Fresh colours exist since $|C| \geq n-1$. Let $c$ be the
first fresh colour that follows, in the cyclic order on $C$, the
colours that are held in $s''$ by players $i \neq 1$.
\III

\NI
\emph{Claim 2.} Player $1$ eventually introduces in $\xi$ the colour $c$.
\II

\NI
\emph{Proof}. 
By the definition of the scheduler and Theorem \ref{thm:fairness}.
\HB
\III

\NI
\emph{Claim 3.} Player $1$ eventually introduces in $\xi$ 
the successor $c'$ of the colour $c$.
\II

\NI
\emph{Proof}. 
By the definition of the scheduler and Theorem \ref{thm:fairness}.
\HB
\II

Consider now the joint strategies $s^1$ and $s^5$ resulting from the
steps described in Claims 2 and 3.  Let
\[
s^4 \to s^5
\]
be the last step of the segment $s^1 \to^* s^5$ of $\xi$.
So $s^1_1 = s^4_1 = s^4_n = c$ and $s^5_1 = c'$. 

Take now a joint strategy $s^6$ from the segment $s^1 \to^* s^5$,
different from $s^1$ and $s^5$. In $s^6$ player $1$ is not
selected. Moreover, by the definition of the game, each better
response of a player different than $1$ is the colour of his
predecessor. So only player $1$ can introduce in $\xi$ colour $c$.

This implies by induction that each time some player $i$ switches in
$s^6$ to the colour $c$, all players $1,\ldots, i-1$ hold in $s^6$ the colour
$c$. So the only possibility that player $n$ holds the colour $c$ in $s^4$
is that all players hold in $s^4$ the colour $c$. Informally, the colour $c$
`travelled the whole ring’. So $s^4$ is a legitimate joint
strategy. Hence by Corollary \ref{cor:2} and Theorem \ref{thm:fairness}
the scheduler $f$ ensures self-stabilization.  
\II

\NI
\emph{Part 2: computing the bound.}

Recall that $s''$ is the first lean joint strategy in $\xi$.  Let $s'$
be the first joint strategy in the segment $s \to^{*} s''$ of $\xi$
such that in the segment $s' \to^{*} s''$ player 1 is not selected.
We first determine the maximum number of steps in the prefix
$s \to^* s'$ of $\xi$.  Since $s''$ is the first lean joint strategy
in $\xi$, in the prefix $s \to^* s'$ only players 1 and 2 are
selected. Moreover, by the choice of $s'$ the last step in this prefix
involves player 1.  Further, player 1 can be selected the second time
only after player $n$ has been selected and no player can be selected
twice in succession.  These constraints leave only two possible
schedulings that yield $s \to^* s'$, namely 1 and 2, 1.

However, the prefix $s \to^* s'$ cannot have 2 steps. Indeed, otherwise it would have
the form
\[
(c_1, c_2, \LL, c_n) \to (c_1, c_1, c_3, \LL, c_n) \to (c'_n, c_1, c_3, \LL, c_n),
\]
where $c_1= c_n$. So $(c_1, c_1, c_3, \LL, c_n)$ is lean, which contradicts
the choice of $s''$ as the first lean joint strategy in $\xi$.
Consequently the prefix $s \to^* s'$ can have at most 1 step.

Let now $\xi'$ be the suffix of $\xi$ that starts in $s'$.  We now
determine the number of steps in $\xi'$ that yield self-stabilization.
We can assume that it takes in $\xi$ at least three steps to reach
$s^5$, as otherwise the bound holds.  Consider the last three steps in
$\xi$ that lead to $s^5$:
\[
s^2 \to s^3 \to s^4 \to s^5.
\]
We noticed already that in $s^4$ all players hold the colour
$c$. Also, the last $n$ steps in $\xi$ that lead to $s^4$ consist of
switching to the colour $c$.  Hence $s^2$ is of the form
$(c, \LL, c, a, b)$, where $a \neq c$ and $b \neq c$.  
\II

\NI
\emph{Case 1} $a = b$.

Then $s^2$ is legitimate.  We first compute the number of steps in the
prefix $\chi$ of $\xi'$ leading from $s'$ to $s^4$.  Consider some
player $i$.  In $\chi$ he can be involved in two types of steps:

\begin{itemize}

\item in which he switches to a colour held in $s'$ by one his
  predecessors $1, \LL, i-1$,

\item in which he switches to a colour introduced in $\chi$ by player
  1 (to identify such steps in $\chi$ we can `mark' such colours in some way).

\end{itemize}

The first possibility leads to at most $i-1$ steps, while the second
one to at most $n-2$ steps since starting from the lean joint strategy $s''$
(and hence from $s'$) player 1 can change his colour in $\chi$ at most
$n-2$ times.  This means that the total number of steps in $\chi$ is
at most
\[
\sum_{i = 1}^n (i - 1 + n-2) = \frac{n (n-1)}{2} + n(n-2).
\]
Deducting 2 for the steps $s^2 \to s^3 \to s^4$
we get the bound $\frac{n (n-1)}{2} + n(n-2) - 2$
on the number of steps in $\xi'$ that yield self-stabilization.
\II

\NI
\emph{Case 2} $a \neq b$.

Then $s^2$ is not legitimate but $s^3$ is, so we need to compute the
number of steps in $\xi'$ leading from $s'$ to $s^3$. To this end we
modify $\xi'$ to another improvement path $\psi$ by replacing the step
$s^2 \to s^3$ by
\[
s^2 \to (c, \LL, c, a, a) \to s^3
\]
and apply the reasoning from Case 1 to $\psi$. This yields the above bound
on the number of steps in $\psi$ needed to reach $(c, \LL, c, a, a)$
and hence the same bound on the number of steps in $\xi'$ leading from
$s'$ to $s^3$.
\III
\II

We noticed already that the prefix $s \to^* s'$ can have at most 1 step, so
we conclude that $\xi$ ensures self-stabilization in
$\frac{n (n-1)}{2} + n(n-2) - 2 + 1 = \frac12 (3 n + 1)(n - 2)$ steps.
\HB
\end{proof}

The original bound of \cite{Dij74} on the number of colours was
$|C| \geq n$.  The authors of \cite{FHP05} noticed that it can be
lowered to $|C| \geq n-1$ and that it is optimal in the sense that for
$|C| = n-2$ the claim of the theorem does not hold.  The latter
observation was established by noting that starting from the joint
strategy
\[
c_2 c_1 c_{n-2} \LL c_2 c_1
\]
the counterclockwise scheduling of the players combined with the
selecting of the colours in the assumed cyclic order by player 1
generates an infinite improvement path which does not yield
self-stabilization.  The fact that self-stabilization can be reached
in $\mathcal{O}(n^2)$ steps when $|C| \geq n$ was established in
\cite{M95}.  Finally, Theorem \ref{thm:1} shows that the use of a
scheduler in Theorem \ref{thm:2} is necessary.

%
%
%
%
Next, we show that $\frac12 (3 n + 1)(n - 2)$
is also a lower bound. 

\begin{example}
  Consider the game $\mathcal{G}$ for $n$ players with $|C| \geq
  n-1$. Assume the cyclic order
  $c_1 \to c_2 \to \dots \to c_{n-1} \to \LL $ on $C$.  So if
  $|C| = n-1$, then $c'_{n-1} = c_{1}$ and otherwise
  $c'_{n-1} = c_{n}$.

  Then the following prefix of an improvement path is generated by every
  scheduler mentioned in Theorem \ref{thm:2} and ends in a legitimate joint
  strategy:
\[
\begin{array}{rcll}
c_{1}c_{n-1}c_{n-2}\ldots c_{1} &\to\\
c_{2}c_{n-1}c_{n-2}\ldots c_{1} & \xrightarrow{n-1 \text{ steps}} & c_{2}c_{2}c_{n-1}c_{n-2}\ldots c_{2} &\to\\
c_{3}c_{2}c_{n-1}c_{n-2}\ldots c_{2} & \xrightarrow{n-1 \text{ steps}} & c_{3}c_{3}c_{2}c_{n-1}\ldots c_{3} &\to\\
c_{4}c_{3}c_{2}c_{n-1}\ldots c_{3} & \xrightarrow{n-1 \text{ steps}} & c_{4}c_{4}c_{3}c_{2}c_{n-1}\ldots c_{4} &\to\\
c_{5}\ldots c_{2}c_{n-1}\ldots c_{4} & \xrightarrow{n-1 \text{ steps}} & c_{5}c_{5}\ldots c_{2}c_{n-1}\ldots c_{5} &\to\\
& \vdots \\
c_{n-1}\ldots c_{2}c_{n-1}c_{n-2} &\xrightarrow{n-1 \text{ steps}} & c_{n-1}c_{n-1}\ldots c_{2}c_{n-1} &\to\\
c'_{n-1}c_{n-1}\ldots c_{2}c_{n-1} &\xrightarrow{\frac{n(n-1)}{2} -2 \text{ steps}} & 
c'_{n-1}c'_{n-1}\ldots c'_{n-1}c_{n-1}c_{n-1}.
\end{array}
\]

The number of steps in the last line needs to be clarified since the
scheduling used in the proof of Note \ref{not:3} yields already after
$\frac{n(n-1)}{2} - (n-1)$ steps the legitimate joint strategy
$c'_{n-1}c_{n-1}\ldots c_{n-1}c_{n-1}c_{n-1}$, so `too
early'. Therefore we modify this scheduling to
\[
(n), (n-1, n), (n-2, n-1, n), \LL, 
(3, 4, \LL, n-1),
(2, 3, \LL n-2, n, n-1, n).
\]
This way we ensure that the legitimate joint
strategy is reached only after $\frac{n(n-1)}{2} - 2$ steps.
Alternatively, we could use the scheduling
\[\
  \begin{array}{l}
(n, n-1, \LL, 2), (n, n-1,\LL,3), \LL, (n, n-1, n-2), (n,n-1), (n), \\
(n, n-1, \LL, 2), (n, n-1,\LL,3), \LL, (n, n-1, n-2), (n).
  \end{array}
\]

The first and the last two lines consist in total of
$1 + \frac{n(n-1)}{2}-2$, so $\frac{(n+1)(n-2)}{2}$ steps, while each
of the remaining $n-2$ lines consists of $n$ steps. Therefore the
total number of steps to reach 
$c'_{n-1}c'_{n-1}\ldots c'_{n-1}c_{n-1}c_{n-1}$ equals
$\frac{(n+1)(n-2)}{2} + n(n-2) = \frac12 (3 n + 1)(n - 2)$.
Note that no other listed joint strategy is legitimate.

\HB
\end{example}

\section{Related work and discussion}

Starting from \cite{HT04}, a paper that relates secret sharing and
multiparty communication protocols to game theory, a growing
literature keeps revealing rich connections between game theory and
distributed computing. For a short overview of the early connections
see Section 4 of \cite{Hal07}.

Let us mention a couple of more recent examples.  The authors of
\cite{ADH13} provide a game-theoretic analysis of the leader election
algorithms on a number of networks for both the synchronous case and
the asynchronous case. In turn, \cite{FO17} provides a framework in
which the processes and the environment of a distributed system are
viewed as players in an extensive game, in which implementations are
interpreted as strategies with an implementation being correct if the
corresponding strategy is winning.

To discuss the papers about connections between game theory and
self-stabiliza\-tion note first that we followed here the original
Dijkstra's definition of a legitimate global state as the one in which
exactly one machine can change its state.  If we view a legitimate
global state as the one in which no machine can change its state and
drop the fairness assumption then we enter the area of
\emph{self-stabilizing algorithms}. An early example of such an
algorithm is the one introduced in \cite{SRR95} that computes a
maximal independent set (MIS).

Probably the first paper that noted the connection between the
self-stabilizing algorithms and game theory is \cite{Dasgupta:2006},
where the notion of a \emph{selfish stabilization} is introduced.  The
authors attached to each node of a graph a cost function (a customary
alternative to the payoff functions in the definition of strategic
games) to derive a simple self-stabilizing algorithm that constructs a
spanning tree in a final state corresponding to a Nash equilibrium
of the underlying strategic game.
In turn, the authors of \cite{Jaggard:2014} related self-stabilization to
\emph{uncoupled dynamics}, a procedure used in game theory to reach a
Nash equilibrium in situations when players do not know each others'
payoff functions.

Recently, the authors of \cite{Yen:2016} observed that
self-stabilizing algorithms that compute a maximal weighted
independent set (MWIS) and MIS can be analyzed using game-theoretic
tools.  To relate this work to ours recall that in our setup we
defined a legitimate joint strategy as the one in which exactly one
player does not play a best response.  Consider now an alternative
definition that equates the legitimate joint strategy with a Nash
equilibrium.  We need now to recall the following definition due to
\cite{MS96}.  We say that a strategic game has the \bfe{finite
  improvement property} (\bfe{FIP}) if every improvement path is
finite.

The authors of \cite{Yen:2016} found that the self-stabilizing
algorithms that compute a MWIS and a MIS correspond to natural
strategic games on graphs that have the FIP. The computations of such
an algorithm then correspond to the (necessarily finite) improvement
paths in the corresponding game.  They also noticed that if a game on
a graph has the FIP then after an appropriate translation to a
distributed system a self-stabilizing algorithm is obtained. Indeed,
the FIP ensures the closure property, while the stability is
immediate. These observations also clarify the set up of the just
discussed papers \cite{Dasgupta:2006} and \cite{Jaggard:2014}.

We conclude this discussion of relations between self-stabilization
and game theory by the following remark.  The author of \cite{Gou01}
introduced the concept of a \bfe{weak self-stabilization} which
guarantees that a distributed system reaches a legitimate state only
by \emph{some} (and thus not necessarily all) sequence of moves. This
concept can be easily incorporated into our framework by stipulating
that a game \bfe{admits weak self-stabilization} if from every initial
joint strategy some improvement path ensures
self-stabilization. Schedulers that ensure self-stabilization
obviously establish weak self-stabilization.  This property naturally
corresponds to the class of weakly acyclic games introduced in
\cite{You93} and \cite{Mil96}. They are defined by the following
weakening of the FIP: a game is \bfe{weakly acyclic} if for every
initial joint strategy there exists a finite improvement path that
starts in it.  For a thorough analysis of weakly acyclic games see
\cite{AS15} from which we adopted the concept of a scheduler.

\section*{Acknowledgment}
We acknowledge useful comments of Mohammad Izadi, Zoi Terzopoulou and
Peter van Emde Boas.  First author was partially supported by the NCN
grant nr 2014/13/B/ST6/01807.

\bibliographystyle{abbrv}
\bibliography{/ufs/apt/bib/s,/ufs/apt/bib/refs}

\end{document}